\documentclass[copyright,creativecommons]{eptcs} \providecommand{\event}{QPL 2019} % Name of the event you are submitting to
\usepackage{underscore}           % Only needed if you use pdflatex.
\usepackage{mymacros}
\usepackage{stmaryrd}

\newcommand{\catname}[1]{{\normalfont\textbf{#1}}}
\newcommand{\emptylist}{( )}

\newcommand{\ghz}{\ket{\text{GHZ}}}

\newcommand{\specialcell}[2][c]{%
  \begin{tabular}[#1]{@{}l@{}}#2\end{tabular}}
  
\newcommand{\pauli}[2]{\sigma_{\MakeLowercase{#1}}^{#2}}

\title{Cohomology and the Algebraic Structure of Contextuality in Measurement Based Quantum Computation}
\author{Sivert Aasn\ae ss
\institute{Department of Computer Science,\\
University of Oxford
}
\email{sivert.aasnaess@cs.ox.ac.uk}
}
\def\titlerunning{Cohomology and the Algebraic Structure of Contextuality in MBQC}
\def\authorrunning{Sivert Aasn\ae ss}
\begin{document}
\maketitle

\begin{abstract}
    Okay, Roberts, Bartlett and Raussendorf recently introduced a 
    new cohomological approach to 
    contextuality in measurement based
    quantum computing.
    We give an abstract description of their obstruction and the algebraic structure
    it exploits, using the sheaf theoretic framework of Abramsky and Brandenburger.
    At this level of generality we contrast their approach to the 
    \v{C}ech cohomology obstruction of Abramsky, Mansfield and Barbosa
    and give a direct proof that \v{C}ech cohomology is at least as powerful.
\end{abstract}

\section{Introduction}

\newcommand{\pmodel}{\sheaf{S}}
Contextuality is a fundamental feature of quantum mechanics
that has been shown to play a central role in certain
models of quantum computing \cite{Howard, CMBQC.5449R}. 
For instance, a result by Raussendorf shows that 
a measurement based quantum computer with mod 2 linear side processing
requires a strongly contextual resource to perform universal computation
\cite{CMBQC.5449R}.

The sheaf theoretic framework of Abramsky and Brandenburger 
describes contextuality using the powerful language of sheaf theory \cite{ABBR}.
One of the insights of this approach is that contextuality in a range of
examples is characterised by the non-vanishing of a cohomological obstruction
that is derived using  \v{C}ech cohomology \cite{CNLC, CCP, gio17}.

More recently Okay et al. described an
obstruction for contextuality in measurement based quantum computation (MBQC) that
is based on group cohomology \cite{TPCQM}.
While the \v{C}ech cohomology obstruction is well defined for any set of quantum
measurements, their obstruction exploits the algebraic
structure of the Pauli measurements used in MBQC. We give a more
abstract account of this approach using
the sheaf theoretic framework. We briefly state our results:

\begin{itemize}
    \item 
    Local (resp. global) value assignments in MBQC induce
    local (resp. global) trivialisations of a sequence 
    \begin{center}
    \begin{tikzcd}
    \zn{2} \arrow[r] &
    X \arrow[r] &
    X/{\zn{2}}
    \end{tikzcd}
    \end{center}
where $X$ is a commutative partial monoid encoding the compositional structure
of commuting measurements.
    \item Mermin's square and GHZ 
        have natural interpretations in terms of this sequence.
    \item Okay et al.'s obstruction can be defined as an obstruction to a 
    local trivialisation
    of a sequence of this form to extend globally.
    \item 
    We give a direct proof that the vanishing of the
    \v{C}ech cohomology obstruction implies the vanishing
    of Okay et al.'s obstruction. 
\end{itemize}

This paper is organised as follows. In Section 2 we review the sheaf theoretic
formulation of quantum contextuality, the \v{C}ech cohomology obstruction,
and the issue of completeness and generalised all-versus-nothing arguments. 
In Section 3 we derive Okay et al.'s obstruction
as a generalisation of the cohomological characterisation of trivial
group extensions. Finally, in section 4,
we apply this obstruction to contextuality and compare it to the \v{C}ech cohomology
obstruction.

\section{Preliminaries}
\paragraph{Sheaf theoretic formulation of contextuality.} 
\iffalse
\begin{definition}[Presheaf.]
A \emph{presheaf} $\sheaf{S}:X^\text{op} \to \catname{C}$ describes a structured
way of attaching data in a category $\catname{C}$ to the open sets of a topological space $X$. 
If we identify $X$ with the category whose objects are open sets $U \subset X$ 
and morphisms are inclusions
$U \subset V$ then $\sheaf{S}$ is simply a contravariant functor from
$X$ to $\catname{C}$. The morphism
$\sheaf{S}(U \subset V):\sheaf{S}(V) \to \sheaf{S}(U)$
associated with an inclusion $U \subset V$ is called a restriction map, if 
$\sheaf{S}$ is implied then we denote this
map as $\resmap{U}{V}$ or simply write $\res{s}{U}$ if $s \in \sheaf{S}(V)$. 
Functoriality means that restriction of data is compatible with the structure of the topological space: 
If $U \subset V \subset W$
then $\resmap{W}{U} = \resmap{V}{U} \circ \resmap{W}{V}$.
\end{definition}
\fi

In the sheaf theoretic approach to contextuality the type of an experiment is described by a 
\emph{measurement scenario} $(X, \mcvx, O)$, where
\begin{itemize}
    \item The set of \emph{measurements} $X$ is a discrete topological space.
    \item The measurement cover $\mcvx \subset \powset{X}$ is a cover of $X$ and 
    furthermore an \emph{anti-chain}
    ($V \subset C \in \mcvx \Rightarrow V \notin \mcvx$). 
    A subset $V \subset X$ is 
    \emph{compatible} if $V \subset C$ for some \emph{context} $C \in \mcvx$. 
    \item $O$ is the set of outcomes.
\end{itemize}
The \emph{event sheaf} $\sheaf{E}:X^\text{op} \to \catname{Set}::V \mapsto O^V$ 
assigns to a set of measurements 
the set of joint outcomes, or \emph{sections}, and restricts a section
$s \in \sheaf{E}(V)$ to a section $\res{s}{U} \in \sheaf{E}(U)$ for $U \subset V$
with function restriction.

The data describing a particular experiment of type $(X, \mcvx, O)$
is specified by an \emph{empirical model}.
Contextuality is often defined in terms of probabilities \cite{bell}. 
We will instead be
concerned with the stronger notion of \emph{possibilistic} 
contextuality \cite{Kochen1975, Mermin}.
A (possibilistic) \emph{empirical model} $\pmodel:(X, \mcvx, O)$ is a subpresheaf
$\pmodel \subset \sheaf{E}:X^\text{op} \to \catname{Set}$
satisfying the conditions
\begin{enumerate}
    \item  $\pmodel(C) \neq \emptyset$ for all $C \in \mcvx$.
    \item  $\pmodel$ is \emph{flasque beneath the cover}: 
        $U \subset V \subset C \in \mcvx \implies 
        \pmodel(U \subset V):\pmodel(V) \to \pmodel(U)$ is surjective.
    \item Every \emph{compatible family} induces a global section: 
    A family ${\{s_C \in \pmodel(C)\}}_{C \in \mcvx}$ is \emph{compatible} 
    if $\res{s_C}{C \cap C'} = \res{s_{C'}}{C \cap C'}$ for all $C,C' \in \mcvx$.
    We require that every compatible 
    family is the family of restrictions 
    of some global section.
\end{enumerate}
We say that $\pmodel:(X, \mcvx, O)$ is
\begin{itemize}
    \item \emph{logically contextual at $s \in \pmodel(C)$} 
    if there is no $g \in \pmodel(X)$ with $\res{g}{C} = s$.
    \item
    \emph{non-contextual} if $\pmodel$ is not logically contextual at any $s \in \pmodel(C)$.
    \item \emph{strongly contextual} if $\pmodel$ is logically contextual at every
    $s \in \pmodel(C)$, equivalently $\pmodel(X) = \emptyset$.
\end{itemize}
\paragraph{Quantum contextuality.} 
If $X$ is a set of Hermitian measurements then we define a measurement scenario
$(X, \mcvx, O)$, where each context $C \in \mcvx$ is a maximal 
subset of mutually commuting measurements and 
$O$ is the combined set of eigenvalues.

A value assignment $s:V \to O$, for a compatible set 
$V = \{M_1, M_2, \cdots, M_n\} \subset X$,
is \emph{consistent with quantum mechanics} if there exists
a state $\ket{\psi}$ such that the joint
outcome specified by $s$ is consistent with $\ket{\psi}$
according to the Born rule
\begin{align*}
    \norm{P_1P_2\cdots P_n\ket{\psi}}^2 \neq 0
\end{align*}
where we require that $s(M_i)$ is an eigenvalue
of $M_i$ and $P_i$ denotes the projector onto the corresponding eigenspace.

The \emph{state independent}, and \emph{state dependent} models 
$\pmodel_X, \pmodel_{X,\psi}:(X, \mcvx, O)$ 
are defined at any $V \subset C \in \mcvx$ \emph{below the cover}
as respectively
\[
    \pmodel_X(V) := 
    \{s:V \to O \mid 
    s \text{ is consistent with quantum mechanics} \}
    \]
and
\[
    \pmodel_{X,\psi}(V) := \{
        s:V \to O \mid s \text{ is consistent with $\ket{\psi}$} \}
     \]
and \emph{above the cover} by the condition that every compatible family
induces a global section. It can be shown that the requirement that 
$\pmodel_{X, \psi}$ and $\pmodel_X$
are
flasque beneath the cover is equivalent to the no-signalling principle \cite{ABBR}.

\begin{definition}
The Pauli $n$-group $P_n$ is the matrix group of $n$-fold tensor products of the Pauli matrices
\begin{align*}
\hspace{-2.0 cm}
    I := \begin{bmatrix}
    1 & 0 \\
    0 & 1 \\
    \end{bmatrix}\quad
    \pauli{X}{} := \begin{bmatrix}
        0 & 1 \\
        1 & 0 \\
    \end{bmatrix} \quad
    \pauli{Y}{} := \begin{bmatrix}
        0 & -i \\
        i & 0 \\
    \end{bmatrix}
    \quad
    \pauli{Z}{} := \begin{bmatrix}
        1 & 0 \\
        0 & -1 \\
    \end{bmatrix}
    \end{align*}
along with multiplicative factors $\pm 1, \pm i$. The elements of $P_n$ with
multiplicative factor $\pm 1$ specify $n$-qubit measurements with outcomes in $\{1, -1 \}$.
As is customary, we identify the groups 
$\{1, -1\} \cong \zn{2}$ and write $\pauli{}{i} \in P_n$, where 
$\pauli{}{} \in \{\pauli{X}{}, \pauli{Y}{},\pauli{Y}{}\}$,
for the  $n$-fold tensor product that is 
$\sigma$ at qubit $i$ and $I$ everywhere else.
\end{definition}

\begin{lemma}
    Let $X \subset P_n$
    be a set of measurements,
    $C \subset X$ a context and $s:C \to \zn{2}$ a value assignment that
    is consistent with quantum mechanics.
    \begin{enumerate}[label=\alph*)]
    \item  $s(M_1M_2) = s(M_1) \oplus s(M_2)$ for all $M_1,M_2 \in C$ such that
    $M_1M_2 \in X$.
    \item If $I \in X$ then $I \in C$ and $s(I) = 0$. Similarly if
    $-I \in X$ then $-I \in C$ and $s(-I) = 1$.
    \end{enumerate}
\end{lemma}
\begin{proof}
    b) is clear.
    For a) let
    $M_1,M_2 \in C$ and take 
    any state $\ket{\psi}$ such that 
    $M\ket{\psi} = s(M)\ket{\psi}$ for all $M \in C$. 
    \[
    s(M_1M_2)\ket{\psi} =  M_1M_2\ket{\psi}
    = M_1(s(M_2)\ket{\psi}) = s(M_1)s(M_2)\ket{\psi} \]
    Hence $s(M_1M_2) = s(M_1) \oplus s(M_2)$ 
    with the identification $\{-1, 1\} \cong \zn{2}$.
\end{proof}

\begin{exmp}[Mermin's square]
Let $\pmodel_X:(X, \mcvx, \zn{2})$ be the state independent model induced by the set of
measurements displayed in
\emph{Mermin's square}

\iffalse
\begin{center}
    \begin{tabular}{c | c | c | c | c }
    & $C_1$ & $C_2$ & $C_3$ & \\
    \hline
    $C_4$ & $\pauli{X}{1}$ & $\pauli{X}{2}$ & $\pauli{X}{1}\pauli{X}{2}$ & $I$ \\
    \hline
    $C_5$ & $\pauli{Z}{2}$ & $\pauli{Z}{1}$ & $\pauli{Z}{1}\pauli{Z}{2}$ & $I$ \\
    \hline
    $C_6$ & $\pauli{X}{1}\pauli{Z}{2}$ & 
    $\pauli{Z}{1}\pauli{X}{1}$ & 
    $\pauli{Y}{1}\pauli{Y}{2}$ & $I$ \\
    \hline
    & $I$ & $I$ & $-I$ & \\
    \end{tabular}
\end{center}
\fi
\begin{center}
    \begin{tabular}{c  c  c  c  c }
    $\pauli{X}{1}$ & $\pauli{X}{2}$ & $\pauli{X}{1}\pauli{X}{2}$ & $I$ \\ &&&& \\
    $\pauli{Z}{2}$ & $\pauli{Z}{1}$ & $\pauli{Z}{1}\pauli{Z}{2}$ & $I$ \\&&&& \\
    $\pauli{X}{1}\pauli{Z}{2}$ & 
    $\pauli{Z}{1}\pauli{X}{2}$ & 
    $\pauli{Y}{1}\pauli{Y}{2}$ & $I$ \\&&&& \\
    
    $I$ & $I$ & $-I$ & \\
    \end{tabular}
\end{center}
Observe that the measurements displayed in any row or column
$M_1, M_2, M_3, M_4$ defines a context and furthermore 
satisfies $M_1 M_2 M_3 = M_4$, where $M_4 = \pm I$.
By Lemma 2.1
any local section $s \in \pmodel(C)$ therefore satisfies one of the following equations
\begin{align}
    \pauli{X}{1} \oplus \pauli{X}{2} \oplus \pauli{X}{1}\pauli{X}{2} &=  0\\ 
    \pauli{Z}{1} \oplus \pauli{Z}{2} \oplus \pauli{Z}{1}\pauli{Z}{2} &=  0\\ 
    \pauli{X}{1} \oplus \pauli{Z}{2} \oplus \pauli{X}{1}\pauli{Z}{2} &=  0\\ 
    \pauli{Z}{1} \oplus \pauli{X}{2} \oplus \pauli{Z}{1}\pauli{X}{2} &=  0\\ 
    \pauli{X}{1}\pauli{Z}{2} \oplus 
        \pauli{Z}{1}\pauli{X}{2} \oplus 
         \pauli{Y}{1}\pauli{Y}{2} &= 
            0\\ 
    \pauli{X}{1}\pauli{X}{2} \oplus 
        \pauli{Z}{1}\pauli{Z}{2} \oplus \pauli{Y}{1}\pauli{Y}{2} &= 
            1 
\end{align}
Any global section $g \in \pmodel_X(C)$ therefore simultaneously satisfies all equations.
However, these equations are mutually inconsistent. Summing together
all of the equations gives $0 = 1$, because each measurement appears in exactly
two equations.
$\pmodel_X$ is therefore strongly contextual.
\end{exmp}
\begin{exmp}[GHZ]
    Let $\pmodel_{X,\text{GHZ}}:(X,\mcvx, \zn{2})$ be the
    state dependent model induced by 
    $\ghz := (\ket{000} + \ket{111})/\sqrt{2}$ and
    $X := \bigotimes_{i=1}^3 \pm \{\pauli{X}{}, \pauli{Y}{}, I\}$.
   
    $\ghz$ is a $+1$-eigenstate of $\pauli{X}{1}\pauli{X}{2}\pauli{X}{3}$
    while it is a $-1$-eigenstate of 
        $\pauli{X}{1}\pauli{Y}{2}\pauli{Y}{3}$, 
        $\pauli{Y}{1}\pauli{X}{2}\pauli{Y}{3}$, and
        $\pauli{Y}{1}\pauli{Y}{2}\pauli{X}{3}$. 
    With the identification $\{-1,1\} \cong \zn{2}$ this means that 
    any global section $g \in \pmodel_{X, \text{GHZ}}(X)$ satisfies the following four equations.
    \[
        \pauli{X}{1}\pauli{X}{2}\pauli{X}{3} = 0, \quad
        \pauli{X}{1}\pauli{Y}{2}\pauli{Y}{3} = 1, \quad
        \pauli{Y}{1}\pauli{X}{2}\pauli{Y}{3} = 1, \quad
        \pauli{Y}{1}\pauli{Y}{2}\pauli{X}{3} = 1
        \]
    By Lemma 2.1 a) $g$ then also satisfies the equations
\begin{align}
    \pauli{X}{1} \oplus \pauli{X}{2} \oplus \pauli{X}{3} 
        &= 0 \\
    \pauli{X}{1} \oplus \pauli{Y}{2} \oplus \pauli{Y}{3} 
        &= 1 \\
    \pauli{Y}{1} \oplus \pauli{X}{2} \oplus \pauli{Y}{3} 
        &= 1 \\
    \pauli{Y}{1} \oplus \pauli{Y}{2} \oplus \pauli{X}{3} 
        &= 1
\end{align}
However, summing together equations (7)-(10) results in $0 = 1$. $\pmodel_{X, \text{GHZ}}$
is therefore strongly contextual.
\end{exmp}
\paragraph{\v{C}ech cohomology.}
Let $\mathcal{U}$ be an open cover of a topological space $X$ and $\sheaf{F}:X^\text{op} \to \catname{AbGrp}$ a presheaf of abelian groups.
The \emph{q-simplices} $\sheaf{N}^q(\mathcal{U})$ of the nerve of $\mathcal{U}$ is the set
of all tuples $\sigma = (U_0, U_1, \cdots, U_q) \in \mathcal{U}^{q+1}$ with non-trivial
overlap $\supp{\sigma} :=  \bigcap_{i=0}^{q} U_i \neq \emptyset$. The $q$-cochains 
$C^q(\mathcal{U}, \sheaf{F}) := \bigoplus_{U \in \mathcal{N}^q(\mathcal{U})} \sheaf{F}(\supp{U})$ is the abelian group of all assignments $\omega$ of a coefficient $\omega(\sigma) \in \sheaf{F}(\supp{\sigma})$ to each simplex $q \in \mathcal{N}^q(\mathcal{U})$ such that $\omega(\sigma) \neq 0$ for
at most finitely many $\sigma$.
Using the notation $\partial_i\sigma$ to denote the $q$ simplex obtained from
a $q+1$ simplex $\sigma$ by omitting the $i$'th element we define for each $q$ a coboundary map $d^q:C^q(\mathcal{U}, \sheaf{F}) \to C^{q+1}(\mathcal{U}, \sheaf{F})$ at each $q$-cochain $\omega$ and $q+1$-simplex $\sigma$ as
\begin{align*}
    d^q(\omega)(\sigma) := 
        \sum_{i=0}^q {(-1)}^{i} 
        \resmap{\supp{\partial_i \sigma}}{\supp{\sigma}}(\omega(\partial_i \sigma))
\end{align*}
where 
$\resmap{U}{V} := 
\sheaf{F}(U \subset V):\sheaf{F}(V) \to \sheaf{F}(U)$.
It can be verified that $d^{q+1} \circ d^q = 0$ and so
\begin{center}
    \begin{tikzcd}
    0 \arrow[r, "d^{-1} := 0" above] &
    C^0(\mathcal{U}, \sheaf{F}) \arrow[r, "d^0" above] &
    C^1(\mathcal{U}, \sheaf{F}) \arrow[r, "d^1" above] &
    C^2(\mathcal{U}, \sheaf{F}) \arrow[r, "d^2" above] &
    \cdots
    \end{tikzcd}
\end{center}
is a cochain complex. The \emph{$q$'th \v{C}ech cohomology group
$H^q(\mathcal{U}, \sheaf{F})$} is the quotient 
$Z^q(\mathcal{U}, \sheaf{F}) / B^q(\mathcal{U}, \sheaf{F})$ of the \emph{$q$-cocycles}
$Z^q(\mathcal{U}, \sheaf{F}) := \ker{d^q}$ over the \emph{$q$-coboundaries}
$B^q(\mathcal{U}, \sheaf{F}) := \im{d^{q-1}}$.

\paragraph{The cohomological obstruction.} 
Suppose now that $\pmodel:(X, \mcvx, O)$ is an empirical model and $s_0 \in \pmodel(C_0)$
is a local section. The cohomological obstruction to $s_0$ lifting to a global section
is defined in terms of the presheaf
$\sheaf{F} := F_\ints \circ \pmodel:X^\text{op} \to \catname{AbGrp}$ of formal linear combinations
of local sections, and two auxiliary presheafs 
\[
 \sheaf{F}_{\tilde{C_0}} :: U \mapsto \ker{\sheaf{F}(U \cap C_0 \subset U)}
 \quad
 \res{\sheaf{F}}{C_0} :: U \mapsto \sheaf{F}(C_0 \cap U)
 \]
\iffalse
If an empirical model $\pmodel:(X, \mcvx, O)$ additionally is a presheaf of
abelian groups then the \emph{relative \v{C}ech cohomology} of $\pmodel$ with
respect to a context $C \in \mcvx$ encodes information about which
local sections at $C$ extend to global sections. However, in practice
this requirement is too strong. The approach of
Abramsky, Mansfield, and Soares Barbosa 
is to instead consider the cohomology of 
 
If an empirical model $\pmodel:(X, \mcvx, O)$ additionally is a presheaf of abelian groups then 
the \emph{relative \v{C}ech cohomology} 
$H_V^*(\mathcal{U}, \sheaf{\pmodel})$ of $\pmodel$ with respect
to $V \subset X$ encodes information about which 
local sections at $V$ extend to global sections. More accurately, there exists
a homomorphism $\gamma:\pmodel(V) \to H^1_V(\mcvx, \pmodel)$ with the property
that $s \in \pmodel(V)$ extends to a global section iff. $\gamma(s) = 0$.

We define $\gamma$ for any presheaf  
$\sheaf{F}:X^\text{op} \to \catname{AbGrp}$ satisfying 
 conditions (1-3) in the definition of an empirical model. 
The \emph{relative \v{C}ech cohomology}
$H^*_V(\mcvx, \sheaf{F}) := H^*(\mcvx, \sheaf{F}_V)$ 
of $\sheaf{F}$
with respect to an open subset $V \subset X$
is the \v{C}ech cohomology of the subpresheaf
$\sheaf{F}_V$
of local sections that vanish on $V$.
\fi
At any $U \subset X$ these presheafs are related to $\sheaf{F}$
by a sequence
\begin{center}
    \begin{tikzcd}
    0 \arrow[r] &
    \sheaf{F}_{\tilde{C_0}}(U) \arrow[r, hookrightarrow] &
    \sheaf{F}(U) \arrow[r, "\resmap{U}{U \cap C_0}"] &
    \res{\sheaf{F}}{C_0}(U) \arrow[r] &
    0
    \end{tikzcd}
\end{center}
which in fact is exact, because $\sheaf{F}$ is flasque beneath the cover.
When lifted to the level of cochain complexes it therefore gives rise to
a short exact sequence
\begin{center}
    \begin{tikzcd}
    0 \arrow[r] &
    C^*(\mcvx, \sheaf{F}_{\tilde{C_0}}) \arrow[r] &
    C^*(\mcvx, \sheaf{F}) 
    \arrow[r] &
    C^*(\mcvx, \res{\sheaf{F}}{C_0}) \arrow[r] &
    0
    \end{tikzcd}
\end{center}
Using standard techniques from homological algebra this short exact sequence
of cochain complexes induces
a \emph{long exact sequence} of cohomology groups
\begin{center}
    \begin{tikzcd}[column sep=small]
    0 \arrow[r] &
    H^0(\mcvx, \sheaf{F}_{\tilde{C_0}}) \arrow[d, phantom, ""{coordinate, name=Z}] \arrow[r] &
    H^0(\mcvx, \sheaf{F}) \arrow[r] &
    H^0(\mcvx, \res{\sheaf{F}}{C_0}) \arrow[dll,
        "\gamma",
            rounded corners,
            to path={ -- ([xshift=2ex]\tikztostart.east)
            |- (Z) [near end]\tikztonodes
            -| ([xshift=-2ex]\tikztotarget.west)
        -- (\tikztotarget)}] \\
    &
    H^1(\mcvx, \sheaf{F}_{\tilde{C_0}}) \arrow[r] &
    H^1(\mcvx, \sheaf{F}) \arrow[r] &
    H^1(\mcvx, \res{\pmodel}{C_0}) \arrow[r] &
    \cdots
    \end{tikzcd}
\end{center}
where $\gamma$ is the connecting homomorphism \cite{weibel}.
Using the identification
$\sheaf{F}(C_0) \cong H^0(\mcvx, \res{\sheaf{F}}{C_0})$ we define
the \emph{obstruction for $s_0$} to extend to a global section to be
$\gamma(1 \cdot s_0) \in H^1(\mcvx, \sheaf{F}_{\tilde{C_0}})$.

\begin{lemma}[\cite{CNLC}]
If the cover $\mcvx$ is \emph{connected}\footnote{
i.e. All pairs $C,C' \in \mcvx$ are connected by a sequence
$C_0 = C, C_1, C_2, \cdots, C_{n-1}, C_n = C'$
with $C_i \cap C_{i+1} \neq \emptyset$.
This assumption is harmless because non-connected components
are completely independent in terms of contextuality.
Incidentally all of the scenarios we will consider are connected.
}
then $\gamma(1 \cdot s_0) = 0$
if and only if $1 \cdot s_0$ extends to a compatible family of
$F_\ints\pmodel$.
\end{lemma}

The obstruction is clearly \emph{sound}. If $g \in \pmodel(X)$ then
$1 \cdot \res{g}{C_0}$ extends to the compatible family 
$\{1 \cdot \res{g}{C}\}_{C \in \mcvx}$. However, in general it is not 
\emph{complete}.
If $\gamma(1 \cdot s) = 0$, then $1 \cdot s$ extends 
to a compatible family
of $F_\ints \pmodel$, but this family might not correspond to any global
section of $\pmodel$. Such a false positive occurs for example
in the case of Hardy's paradox \cite{CCP,gio17, gio18}.

\paragraph{Generalised AvN models.} 
Examples 2.1 and 2.2 illustrate a type of contextuality proof that Mermin called 
`all versus nothing' \cite{Mermin2}.
These proofs can be understood as exhibiting an inconsistent set of equations
over $\zn{2}$ that is locally satisfied by the model.
The \v{C}ech cohomology obstruction is complete for the \emph{generalised AvN models},
the class of models that locally satisfies a system of inconsistent equations over any ring $R$ \cite{CCP}.
Let $R$ be a ring and suppose that  
$\pmodel:(X,\mcvx, R)$ is an empirical model.
An \emph{$R$-linear equation} $\phi$ at a context $C \in \mcvx$ 
is a formal sum
\[
    \sum_{x \in C} \mathbf{r}(x)x = a \]
where $\mathbf{r}:C \to R$ and $a \in R$.
A local section $s:C \to R$ \emph{satisfies} 
$\phi$, written $s \models \phi$,
if $\sum_{x \in C} \mathbf{r}(x) \cdot s(x) = a$,
where $\cdot$ denotes multiplication in $R$. 
The \emph{$R$-linear theory of $\pmodel$} is the set of all $R$-linear equations 
that are consistent with $\pmodel$.
\[
    \text{Th}_R(\pmodel) := 
        \bigcup_{C \in \mcvx} 
        \{\phi \text{ is an $R$-linear equation at $C$} \mid s \models \phi, \forall s \in \pmodel(C)\}
\]
\begin{definition}
$\pmodel$ is $\text{AvN}_R$ if its $R$-linear theory is \emph{inconsistent}.
i.e. there is no $s:X \to R$
such that $\res{s}{C} \models \phi$, for every context $C \in \mcvx$
and formula $\phi \in \text{Th}_R(\pmodel)$ at $C$.
\end{definition}
\begin{theorem}[\cite{CCP}]
If $\pmodel$ is $\text{AvN}_R$ then $\gamma(1 \cdot s) \neq 0$ for all
$C \in \mcvx$ and $s \in \pmodel(C)$.
\end{theorem}

\section{The group cohomology obstruction}
If $G$ is a commutative group and $H \leq G$ is a subgroup then
it is not always the case that $G \cong H \times G/H$.
More generally, 
if $H \leq K \leq G$ then a \emph{local trivialisation}
$\phi:K \cong  H \times K/H$
might not arise as a restriction of any
\emph{global trivialisation} $\phi':G \cong H \times G/H$.
In group cohomology the local trivialisations that can be
extended globally are characterised by a vanishing cohomological obstruction
\cite{brown, webb}.
The obstruction of Okay et al. can be understood
as a natural generalisation of this idea to the case where
$G$ is a commutative partial monoid.

\begin{definition} Let $A$ be a commutative group, $X$ a commutative partial monoid,
and $i:A \to X$ an injective homomorphism. 
Consider the sequence $A \xrightarrow{i} X \xrightarrow{\pi} X/A$, where
$\pi:X \to X/A$ is the canonical quotient of the group action 
$l_A:A \times X \to X::(a,x) \mapsto i(a) + x$. \footnote{
Observe that $i(a) + x$ is always defined, even when $X$ is partial,
because $0 + x = i(-a) + (i(a) + x)$
}
\begin{itemize}
    \item A \emph{left splitting} is a homomorphism $s:X \to A$ such that 
    $s \circ i = \id{A}$.
    \item A \emph{right splitting} is a homomorphism $h:X/A \to X$ such 
    that $\pi \circ h = \id{X/A}$.
    \item \specialcell{A \emph{trivialisation} is
    a homomorphism\\
    $\phi:X \to A \times X/A$ such that
    the following diagram commutes:}
        \begin{tikzcd}
        A \arrow[r, "i"] \arrow[dr, "\text{in}_1" below left]&
        X \arrow[r, "\pi"] \arrow[d, "\phi"] &
        X/A \\
        & A \times X/A \arrow[ur, "\text{proj}_2" below right] &
        \end{tikzcd}
        \\
Where $\times$ denotes the cartesian product, and 
$\text{in}_1$, $\text{proj}_2$ refers
to the associated inclusion and projection maps respectively.
\end{itemize}
\end{definition}
\iffalse
We first note that 
Let $A$ be a commutative group, 
$X$ a commutative partial monoid, and $i:A \to X$
an injective homomorphism. Note that
$i(a) + x$ is defined for every $a \in A$ and $x \in X$ 
because $0 + x = i(-a) + (i(a) + x)$. The \emph{quotient partial
monoid $X/A$} is the orbit space
of the group action 
    $l_A:A \times X \to X::(a,x) \mapsto i(a) + x$
with the unique partial monoid structure determined by the canonical quotient
map $\pi:X \to X/A$.
\fi

In this section we will show that under the assumption that $l_A$ is 
free 
the group cohomology obstruction can be generalised to an obstruction
for a \emph{local} trivialisation $\phi:C \to A \times C/A$, where
$i(A) \subset C \subset X$ is a submonoid, to extend globally. We first
translate the problem into one about right splittings.

\paragraph{The splitting lemma.}
It follows from a general fact about the cartesian product \cite{cpp-notes} that
the maps
\[
    \phi \mapsto \text{proj}_1 \circ \phi,
    \quad
    s \mapsto \langle s, \pi \rangle
\]
where $\langle s,\pi \rangle := x \mapsto (s(x),\pi(x))$,
defines a bijective
correspondence between left splittings and trivialisations.
Because this correspondence is compatible with restrictions, 
the problem of extending a trivialisation is equivalent to
the problem of extending a left splitting.
When $l_A$ is free something similar is true about right splittings.
\begin{lemma}
    Any trivialisation $\phi$ is in fact an isomorphism
\end{lemma}
\begin{proof}
    Write $\phi_1 := \text{proj}_1 \circ \phi$
    and define
    $\inv{\phi}(a, [x]) := x + i(a - \phi_1(x))$.
    This is well defined independently of the representative $x$ because
    $\phi_1$ is a splitting.
    Using
    $\phi = \langle \phi_1, \pi \rangle$ and the properties of left
    splittings it is straightforward to verify that $\inv{\phi}$ is both
    a left and right inverse to $\phi$.
\end{proof}

\begin{lemma}[Splitting lemma]
    Suppose that $l_A$ is free, $i(A) \subset C \subset X$ is a submonoid,
    and that $\phi:C \to C \times C/A$ is a trivialisation.
    The following conditions are then equivalent.
\begin{enumerate}
    \item There exists a left splitting $s:X \to A$ such that 
    $\res{s}{C} = \text{proj}_1 \circ \phi$.
    \item There exists a right splitting $h:X/A \to X$ such that
    $\res{h}{C/A} = \inv{\phi} \circ \text{in}_2$.
    \item There exists a trivialisation $\phi':X \to A \times X/A$
    such that $\res{\phi'}{C} = \phi$.
\end{enumerate}
\end{lemma}
\begin{proof}
    $1. \Leftrightarrow 3.$ $\phi \mapsto \text{proj}_1 \circ \phi$ is a bijection
    between trivialisations and left splittings and furthermore
    compatible with restrictions: 
        $\res{s'}{C} =  s  \iff \res{\phi'}{C} = \phi$ whenever $\phi \mapsto s$,
        and $\phi' \mapsto s'$.\\
    $2 \Leftrightarrow 3.$ We show that the map
        $\phi \mapsto \inv{\phi} \circ \text{in}_2$ from trivialisations
        to right splittings, has a left inverse.
        For any right splitting
    $h:X/A \to X$ let
    $\Phi(h) := \langle s, \pi \rangle$ where $s:X \to A$ is defined by the equation
    \[
        h(\pi(x)) = x - i(s(x)) \]
    which has a unique solution because $l_A$ is free. To see that $s$
    in fact is a splitting note first that
    $h(\pi(x+y)) = h(\pi(x)) + h(\pi(y))$ and hence
    \[
        x + y - i(s(x + y)) = x - i(s(x)) + y - i(s(y)) \]
    because $s$ is unique we therefore have $s(x+y) = s(x) + s(y)$. 
    For $s \circ i = \id{A}$ we have
    $h(\pi(i(a))) = h(\pi(0))$, therefore by uniqueness we have $s(i(a)) = a$.
    Finally,
    \begin{align*}
        (\inv{\phi} \circ \text{in}_2)(\pi(x)) &= \inv{\phi}(0, \pi(x)) \\
        &= x + i(0 - \phi_1(x)) \\
        &= x - i(\phi_1(x))
        \end{align*}
    hence by uniqueness $\Phi(\inv{\phi} \circ \text{in}_2) = \phi$,
    and so $\Phi$ is a left inverse of $\phi \mapsto \inv{\phi} \circ \text{in}_2$. 
    Because the map is defined pointwise it is clear that it
    is compatible with restrictions.
\end{proof}

\paragraph{An obstruction to global trivialisations.}

\begin{definition} We define the 
\emph{relative cohomology groups $H^*(M,N;G)$} 
of commutative
partial monoids $N \subset M$ with coefficients in an abelian group $G$.
For each $n \geq 0$ let $M_n$, and similarly $N_n$, be defined by
$M_0 = \{\emptylist\}$
and for $n > 0$
\[
	M_n := \{(m_1, m_2, \cdots, m_n) \in M^{n} \mid m_1 + m_2 + \cdots + m_n \text{ is defined} \}
	\]
The \emph{relative $n$-cochains}
$C^n(M,N;G) := \{f:M_n \to G \mid \res{f}{N_n} = 0\}$
is the abelian group of functions from $M_n$ to $G$ that vanish on $N_n$,
and the coboundary maps
\begin{center}
	\begin{tikzcd}
	0 = C^0(M,N;G) \arrow[r, "d^0" above]& 
	C^1(M,N;G) \arrow[r, "d^1" above]& 
	C^2(M,N;G) \arrow[r, "d^2" above]& 
	\cdots
	\end{tikzcd}
\end{center}
are given by
\begin{align*}
    d^n(f)(m_0,m_1, \cdots, m_n) := 
    &f(m_1, \cdots, m_n) \\
    + &\sum_{i=1}^n {(-1)}^i f(m_0, \cdots, m_{i-1}, m_i + m_{i+1}, m_{i+2}, \cdots, m_n) \\
    + &{(-1)}^{n}f(m_0, \cdots, m_{n-1})
    \end{align*}
$H^n(M,N;G) := Z^n(M,N;G) / B^n(M,N;G)$ is defined as the quotient of the
\emph{relative $n$-cocycles} $Z^n(M,N;G) := \ker{d^n}$ over the \emph{relative $n$-coboundaries} 
$B^n(M,N;G) := \im{d^{n-1}}$.
Note that $\res{f}{N_n} = 0 \implies \res{d^n(f)}{N_{n+1}} = 0$. It can also be
shown that $d^{n+1} \circ d^n = 0$. However, for our purpose it
is sufficient to check this for the maps
\begin{align}
    \label{defd2}
	d^2(f)(m_1,m_2,m_3) &= 
	    f(m_2,m_3) - f(m_1 + m_2,m_3) + f(m_1, m_2 + m_3) - f(m_1,m_2)\\
	d^1(f)(m_1,m_2) &= f(m_2) - f(m_1+m_2) + f(m_1)\\
	d^0 &= 0
\end{align}
which is easily done.
\end{definition}
Suppose now that $l_A$ is free and that $\phi:C \to A \times C/A$ is a 
trivialisation for some submonoid $i(A) \subset C \subset X$.
By the splitting lemma we can equivalently consider
the splitting $R(\phi) := \inv{\phi} \circ \text{in}_2:C/A \to C$.

\begin{definition}
Let $\eta:X/A \to X$ be any choice of representatives that coincides with
$R(\phi)$ on $C/A$.
The \emph{cohomological obstruction} to $\phi$ is the cohomology
class $[\beta] \in H^2(X/A,C/A;A)$ of $\beta$, where $\beta \in Z^2(X/A,C/A;A)$
is uniquely defined by 
    \begin{equation} \label{defbeta}
    	\eta(q_1 + q_2) = \eta(q_1) + \eta(q_2) + i(\beta(q_1,q_2))
    \end{equation}
    for all $q_1,q_2 \in X/A$ with $q_1 + q_2$ defined.
\end{definition}

\begin{lemma}
    The obstruction is well defined and independent of the choice of representatives.
\end{lemma}
\begin{proof}
    First note that $\beta$ is unique because $l_A$ is free, and a relative cochain
    because $\res{\eta}{C/A} = R(\phi)$ is a homomorphism.
    Next,
    to show that $\beta$ is a cocycle we
    use~\eqref{defbeta} and associativity to 
    expand $\eta(q_0 + q_1 + q_2)$
    as both
	\[
		\eta(q_1) + \eta(q_2) + \eta(q_3) + 
		    i(\beta(q_1, q_2+q_3) + \beta(q_2, q_3)) 
    \]
	and
	\[
		\eta(q_1) + \eta(q_2) + \eta(q_3) + 
		    i(\beta(q_1 + q_2, q_3) + \beta(q_1, q_2)) \]
	Because these terms are equal and $l_A$ is free
	\[
	\beta(q_1, q_2 + q_3) + \beta(q_2,q_3) = \beta(q_1 + q_2, q_3) + \beta(q_1,q_2)
	\]
	Comparing this to~\eqref{defd2} gives $d^2(\beta) = 0$, as required.
	Finally, to see that $[\beta]$ is independent of the choice of representatives
	suppose that we instead chose
	$\eta' := \eta + i \circ \gamma$, for some
	$\gamma \in C^1(X/A,C/A;A)$, and
	similarly defined $\beta'$. Expanding~\eqref{defbeta} in the case of
	$\eta'$ in terms of $\eta$ and $i \circ \gamma$ gives
	 \[
	    \eta(q_1 + q_2) = \eta(q_1) + \eta(q_2) + i(\beta'(q_1,q_2) +
	    \gamma(q_1) - \gamma(q_1+q_2) + \gamma(q_2)) \]
	By uniqueness we therefore have $\beta = \beta' + d^1(\gamma)$
	and hence $[\beta] = [\beta']$.
\end{proof}
\begin{theorem}
    The following conditions are equivalent.
    \begin{enumerate}
        \item There exists a trivialisation $\phi':X \to A \times X/A$ such
        that $\res{\phi'}{C} = \phi$.
        \item $[\beta] = 0$
    \end{enumerate}
\end{theorem}
\begin{proof}
$1. \iff$ There exists a right splitting $h:X/A \to X$
such that $\res{h}{C/A} = R(\phi)$.\\
$\iff$ There exists $\gamma \in C^1(X/A,C/A;A)$ such that $\eta + i \circ \gamma$
is a homomorphism.\\
$\iff$ There exists $\gamma \in C^1(X/A,C/A;A)$ such that for all
$q_1, q_2 \in M/A$ with $q_1 + q_2$ defined
\[
	\eta(q_1 + q_2) = 
	\eta(q_1) + \eta(q_2) + i(\gamma(q_1) - \gamma(q_1 + q_2) + \gamma(q_2))
	\]
$\iff$ There exists $\gamma \in C^1(X/A,C/A;A)$ such that
$\beta = d(\gamma)$\\
$\iff$ $[\beta] = 0 \in H^2(X/A,C/A;A)$.
\end{proof}

\section{The group cohomological approach to contextuality}
Suppose that 
$X \subset \bigotimes_{i=1}^n \pm \{\pauli{X}{},\pauli{Y}{},\pauli{Z}{},I\}$ 
is a set of Pauli measurements satisfying the two conditions
\begin{enumerate}
    \item $\{I,-I\} \subset X$.
    \item $M_1, M_2 \in X$ and $M_1M_2=M_2M_1 \implies M_1M_2 \in X$.
\end{enumerate}
In this case matrix multiplication gives each context 
$C \in \mcvx$ the structure of a commutative monoid
and the 
embedding $i:\zn{2} \to C::k \mapsto {(-1)}^k I$ induces a sequence
\begin{center}
    \begin{tikzcd}
    \zn{2} \arrow[r, "i"] &
    C \arrow[r, "\pi"] &
    C/\zn{2}
   \end{tikzcd} 
\end{center}
By Lemma 2.1 every
$s:C \to \zn{2}$ that is consistent with quantum mechanics is a left splitting. 
It follows
that both the state independent and state dependent models 
$\pmodel_X, \pmodel_{X,\psi}:(X, \mcvx, \zn{2})$ are instances
of the following definition.

\iffalse
The obstruction of Okay et al. is not defined for a set of quantum
measurements satisfying
but instead exploits the algebraic
properties of the Pauli measurements.
Formulated abstractly in the sheaf theoretic framework these properties are
summarised by Definition 4.1.
\fi
\begin{definition}
In this section we will assume that we are working with an
empirical model
$\pmodel:(X, \mcvx, A)$ with the additional structure:
\begin{enumerate}
	\item The set of outcomes is a commutative group $(A, +_A, 0_A)$. 
	\item Each $C \in \mcvx$ is a commutative 
		monoid $(C, +_C, 0_C)$ and the monoid structures on different contexts are
		compatible. For all $C,C' \in \mcvx$:
		\begin{enumerate}
		    \item $0_C = 0_{C'}$.
		    \item $x,y \in C \cap C' \implies x +_C y = x +_{C'} y$.
		\end{enumerate}
	\item We are given an embedding $i:A \to \bigcap_{C \in \mcvx} C$ such that
	for each $C \in \mcvx$ the action
	$A \times C \to C::(a,x) \mapsto i(a) +_C\ x$ is free.
	\item Every local section $s \in \pmodel(C)$ is a left splitting
	of the sequence
        \begin{tikzcd}
        A \arrow[r, "i"] &
        C \arrow[r, "\pi"] &
        C/A
        \end{tikzcd}.
\end{enumerate}
\end{definition}
\iffalse
Before we go on to investigate the consequences of this structure
for the characterisation of contextuality with cohomology we explain
what it means concretely in the case of Pauli measurements.

\begin{lemma}
Let $\ket{\psi}$ be an n-qubit state
and
$X \subset \bigotimes_{i=1}^n \pm \{\pauli{X}{},\pauli{Y}{},\pauli{Z}{},I\}$ 
a set of $n$-qubit Pauli measurements with the property that
\begin{enumerate}
    \item $\{I,-I\} \subset X$.
    \item $M_1, M_2 \in X$ and $M_1M_2=M_2M_1 \implies M_1M_2 \in X$.
\end{enumerate}
The state independent
and state dependent models
$\pmodel_X:(X, \mcvx,\zn{2})$ and $\pmodel_{X,\psi}:(X, \mcvx, \zn{2})$
satisfies Definition 4.1.
\end{lemma}
\begin{proof}
Each context $C \in \mcvx$ is closed under matrix
multiplication and is therefore a commutative monoid with identity $I$.
Let $i$ be the homomorphism
\[
    i(k) = {(-1)}^k I, \quad k \in \zn{2} \]
That the local sections are homomorphisms and satisfies $s \circ i = \id{\zn{2}}$
follows from  Lemma 2.1 a) and b) respectively. For every $C \in \mcvx$
$\pmodel_{X,\psi}(C) \subset \pmodel_{X}(C)$. Definition 3.1 therefore also
holds for
$\pmodel_{X,\psi}$.
\end{proof}
\fi
\paragraph{Group cohomology} 
The monoid structures on different contexts are compatible
and therefore ``glue together'' to define a 
commutative \emph{partial} monoid $(X, +, 0)$ whose maximal
submonoids correspond to the contexts. We consider the sequence
\begin{equation} \label{global}
    \begin{tikzcd}
    A \arrow[r, "i"] &
    X \arrow[r, "\pi"] &
    X/A
    \end{tikzcd}
\end{equation}
induced by $i:A \to \bigcap_{C \in \mcvx} C$ and note that the
action $l_A:A \times X \to X::(a,x) \mapsto i(a) + x$ is free. 
Suppose now that $C \in \mcvx$ is a particular context and $s \in \pmodel(C)$
a local section. Because $s$ is a splitting
it induces a \emph{local} trivialisation
\begin{center}
    \begin{tikzcd}
    A \arrow[r, "i"] \arrow[dr, "\text{in}_1" below left]&
    C \arrow[r, "\pi"] \arrow[d, "{\langle s,\pi \rangle}"] &
    C/A  \\
    & A \times C/A \arrow[ur, "\text{proj}_2" below right] &
    \end{tikzcd}
\end{center}
of sequence~\eqref{global}.
\begin{definition}
    $[\beta_s] \in H^2(X/A,C/A;A)$ is the
    cohomological obstruction to the existence of a trivialisation
    of sequence~\eqref{global} that extends 
    $\langle s,\pi \rangle:C \to A \times C/A$.
\end{definition}

The obstruction is clearly \emph{sound}. A global section
$g \in \pmodel(X)$ is a splitting because it's restriction to every context
is a splitting. If furthermore $\res{g}{C} = s$ then $\langle s, \pi \rangle$
extends to $\langle g, \pi \rangle: X \to A \times X/A$, by
Theorem 3.1 we therefore have $[\beta_s] = 0$.

\paragraph{Proofs of contextuality.}
Although the obstruction is sound, it is not in general \emph{complete}.
False positives can arise in the form of global extensions $\langle g, \pi \rangle$ of
$\langle s, \pi \rangle$ that correspond to splittings $g \notin \pmodel(X)$ that are not
allowed by $\pmodel$.
For this purpose Okay et al. introduced `topological' versions of Mermin's square and GHZ.
These proofs can be understood as showing that there are no false positives in the form
of right splittings. We note however, 
that the original proofs almost exactly spells out that there are 
no false positives in the form of left splittings.

\begin{exmp}[Mermin's square]
Let $X \subset P_2$ be any set of Pauli measurements that is closed under products
of commuting measurements
and contains the measurements displayed in Mermin's square.
We consider the state independent model 
$\pmodel_X:(X,\mcvx, \zn{2})$ which in this case satisfies Definition 4.1.

Observe that equations (1)-(6) induced by Mermin's square all can be rearranged
to be on the form
\[
    M_1 \oplus M_2 =  M_1 M_2 \]
for $M_1,M_2 \in X$ with $M_1M_2=M_2M_1$. That the equations
are mutually inconsistent therefore literally says that there is no
global left splitting. We therefore have $[\beta_s] \neq 0$ for every local section $s$ of $\pmodel$.
\end{exmp}

\begin{lemma}
    Suppose that $X \subset P_n$ is a set of Pauli measurements that contains the identity
    and is closed under commuting products. For any state $\ket{\psi}$
    the set of measurements
    \[
        X_\psi := \{M \in X \mid M\ket{\psi} = \pm \ket{\psi} \}
        \]
    whose outcome is uniquely determined by $\ket{\psi}$ is a submonoid
    of $X$. 
\end{lemma}
\begin{proof}
    Because the Pauli measurements
    $\pauli{X}{}, \pauli{Y}{}, \pauli{Z}{}$ pairwise \emph{anti-commute}
    \[
        \pauli{X}{}\pauli{Z}{} = -\pauli{Z}{}\pauli{X}{}, \quad
        \pauli{X}{}\pauli{Y}{} = -\pauli{Y}{}\pauli{X}{}, \quad
        \pauli{Y}{}\pauli{Z}{} = -\pauli{Z}{}\pauli{X}{} \]
    all $M_1,M_2 \in X$ either commute, or anti-commute. The condition
    that $M_1,M_2 \in X_\psi$ forces the former.
    Furthermore note
    that $X_\psi$ contains $I$ and is closed under products.
\end{proof}

\begin{exmp}[GHZ]
\label{example2}
Let $X := \bigotimes_{i=1}^3 \pm \{\pauli{X}{}, \pauli{Y}{}, \pauli{Z}{}, I\}$. First note
that the state dependent model
$\pmodel_{X,\text{GHZ}}:(X, \mcvx, \zn{2})$ is an instance of Definition 4.1 because $X$
is closed under commuting products and contains $\pm I$.
Next, consider the set $X_\text{GHZ}$ of measurements whose outcome is uniquely determined
by $\ghz$ and observe that
equations (7-10) in Example 2.2
are all of the form
\[
    M_1 \oplus M_2 \oplus M_3 = s_{\text{GHZ}}(M_1M_2M_3) \]
where 
$M_1,M_2,M_3 \in X$ are compatible, $M_1M_2M_3 \in X_{\text{GHZ}}$,
and $s_{\text{GHZ}}(M_1M_2M_3)$ is the unique outcome that is consistent with $\ghz$.
That the equations are mutually inconsistent therefore ensures that there is no
global splitting $g:X \to \zn{2}$ whose restriction to $X_\text{GHZ}$ is
$s_\text{GHZ}$.
It follows that if $C \in \mcvx$ is any context 
that contains $X_{\text{GHZ}}$ then $[\beta_s] \neq 0$ for every 
$s \in \pmodel_{X, \text{GHZ}}(C)$.
Note that such a context exists because the maximal submonoids of $X$ are the contexts
and by Lemma 4.1 $X_\text{GHZ}$ is a monoid.
\end{exmp}

\paragraph{Comparison with \v{C}ech cohomology.} 
The \v{C}ech cohomology obstruction is defined for
all empirical models, but this generality comes at a price. It is not
a complete characterisation of contextuality.
It is therefore natural to ask if there are any examples of contextuality
that is detected by group cohomology, but not \v{C}ech cohomology.
Because Mermin's square and GHZ are examples of all versus nothing arguments we
know that \v{C}ech cohomology
detects contextuality in both cases. We now show more generally that 
if the group cohomology obstruction is non-trivial, then the \v{C}ech cohomology
obstruction is also non-trivial.

\begin{theorem}
    Let $s_0 \in \pmodel(C_0)$ be any local section. Then
    \[
        \gamma(1 \cdot s_0) = 0 \implies [\beta_{s_0}] = 0 \]
\end{theorem}
\begin{proof}
    First note that it follows from 
    $i(A) \subset \bigcap_{C \in \mcvx} C$ that the cover $\mcvx$ is connected.
    Therefore, if $\gamma(1 \cdot s_0) = 0$ then there is some 
    compatible family
    $\{r_C \in F_\ints \pmodel(C)\}_{C \in \mcvx}$ such that
    $r_{C_0} = 1 \cdot s_0$. 
    Observe now that any such family in fact is a
    compatible family of
    formal affine combinations: For any $C \in \mcvx$
    \[
    \sum_{s \in \pmodel(C)} 
            r_C(s) \cdot \res{s}{C \cap C_0}
            =
        \res{r_C}{C \cap C_0} = 
        \res{r_{C_0}}{C \cap C_0} = 1 \cdot \res{s_0}{C \cap C_0}
        \]
    hence $\sum_{s \in \pmodel(C)} r_C(s) = 1$.
    We now use the unique module
    action\footnote{
    i.e. $0 \cdot a = 0$, and for $n \geq 1$: 
    $n \cdot a := a + a + \cdots + a$ ($n$ times) and $-n \cdot a = -(n \cdot a)$.
    }of $\ints$ on $A$ to collapse this
    compatible family to a function $g:X \to A$.
    \[
        g(x) := \sum_{s \in \pmodel(C)} r_C(s) \cdot s(x), \quad \text{where
        $C \in \mcvx$ is any context with $x \in C$}
        \]
    Because the set of splittings is closed under affine combinations
    this function is in fact a splitting which furthermore extends $s_0$.
    We therefore have $[\beta_{s_0}] = 0$.
\end{proof}
\section{Conclusion}
We have considered two different applications of cohomological techniques to
contextuality in MBQC. While the \v{C}ech
cohomology obstruction is defined for any set of quantum measurements, the group
cohomology obstruction relies on the specific algebraic 
structure of the  Pauli measurements. 
We have given an abstract account of this approach 
using the sheaf theoretic framework.
At this level of generality we observe that
although both approaches rely on structural assumptions to be complete, there
is a direct way in which the \v{C}ech cohomology obstruction subsumes the group
cohomology obstruction.

Our presentation of the group cohomology approach deviates from
Okay et al.'s in that we have defined a single obstruction that applies to
both state independent and state dependent contextuality. We have shown that this
obstruction detects contextuality in the state independent
case of Mermin's square and the state dependent case of GHZ.

\paragraph{Acknowledgements.} I would like to thank Samson Abramsky,
Rui Soares Barbosa and Giovanni Car\`u for their guidance and 
valuable discussions.
Support from the Aker Scholarship and 
Sparebank 1 Ringerike Hadeland's Talentstipend is also gratefully
acknowledged.
\nocite{*}
\bibliographystyle{eptcs}
\bibliography{references}
\end{document}